\documentclass[letterpaper, 10 pt, conference]{ieeeconf}  % Comment this line out if you need a4paper

\IEEEoverridecommandlockouts                              % This command is only needed if 
\usepackage{graphics} % for pdf, bitmapped graphics files
\usepackage{epsfig} % for postscript graphics files
\usepackage{mathptmx} % assumes new font selection scheme installed
\usepackage{times} % assumes new font selection scheme installed
\usepackage{amsmath} % assumes amsmath package installed
\usepackage{amssymb}  % assumes amsmath package installed
\usepackage{algorithm}
\usepackage[noend]{algpseudocode}
\usepackage{comment}
\usepackage{pgfplots}
\usepackage{cite}
\usepackage{dsfont}
\usepackage{xcolor}

\usepackage{tikz}
\usetikzlibrary{shapes.geometric, arrows, positioning, calc}
\definecolor{pinegreen}{rgb}{0.004, 0.475, 0.435}

\newtheorem{definition}{Definition}
\newtheorem{lemma}{Lemma}
\newtheorem{theorem}{Theorem}
\newtheorem{remark}{Remark}

\newtheorem{corollary}{Corollary}

\DeclareMathOperator*{\argmax}{arg\,max}

\newcommand{\qedwhite}{\hfill \ensuremath{\Box}}

\title{\LARGE \bf
Aggregate Fictitious Play for Learning in Anonymous Polymatrix Games (Extended Version)
}

\author{Semih Kara and Tamer Ba\c{s}ar% <-this % stops a space
\thanks{The authors are with the Coordinated Science Laboratory, University of Illinois Urbana-Champaign, Urbana, USA, 61801; Emails: {\tt\small \{semihk, basar1\}@illinois.edu}}%
\thanks{This work was supported in part by the Army Research Office (ARO) Grant W911NF-24-1-0085.}% <-this % stops a space
}

\begin{document}

\maketitle
\thispagestyle{empty}
\pagestyle{empty}

\begin{abstract}

Fictitious play (FP) is a well-studied algorithm that enables agents to learn Nash equilibrium in games with certain reward structures. However, when agents have no prior knowledge of the reward functions, FP faces a major challenge: the joint action space grows exponentially with the number of agents, which slows down reward exploration. Anonymous games offer a structure that mitigates this issue. In these games, the rewards depend only on the actions taken; not on who is taking which action. Under such a structure, we introduce aggregate fictitious play (agg-FP), a variant of FP where each agent tracks the frequency of the \textit{number of other agents} playing each action, rather than these agents' individual actions. We show that in anonymous polymatrix games, agg-FP converges to a Nash equilibrium under the same conditions as classical FP. In essence, by aggregating the agents' actions, we reduce the action space without losing the convergence guarantees. Using simulations, we provide empirical evidence on how this reduction accelerates convergence.

\end{abstract}

\section{Introduction} \label{sec:Intro}

Learning in multi-agent systems is a cornerstone of game theory \cite{fudenberg_levine_th-of-learning-in-games, zhang-et-al_marl-overview, ozdaglar-et-al_independent_learning-stochastic-games}. Among the many learning dynamics, fictitious play (FP) \cite{brown_some-notes-on-comp-of-games-sols} has attracted significant attention due to its simplicity and capacity to lead agents to Nash equilibria in games with certain reward structures \cite{robinson_an-iter-method-for-solving-a-game,monderer_shapley_pot-games,Fict_play_in_nets_Ewerhart_Valkanova}. In classical FP, agents \textit{repeatedly play a fixed matrix game} \cite{osborne-rubinstein_game-theory}, form \emph{beliefs} about other agents’ strategies, and choose actions that \emph{maximize} their expected rewards. 

To implement FP, each agent needs to know its reward function and observe the actions of other agents. Both of these assumptions are relaxed in \cite{Conv_Multi_Timescale_RL_Algos_in_NFGs_Leslie_Collins}, where the authors combine smooth FP with a Q-learning-inspired reward estimation process that runs on a faster timescale. 
% Their method works in radically uncoupled settings, where agents have no knowledge of reward functions and cannot observe others' actions \cite{ozdaglar-et-al_independent_learning-stochastic-games}. 
In follow-up work, the same authors propose individual Q-learning \cite{Ind_Q_Learning_in_NFGs_Leslie_Collins}, which simplifies the setup to a single timescale. A related two-timescale method is presented in \cite{sayin-et-al_fict-play-in-zero-sum-stoch-games}, where FP is combined with Q-learning. Their setting is \textit{model-free}: agents do not know the reward functions but can observe each others' actions, allowing them to form beliefs. Building on this, \cite{donmez-sayin_ind-q-learning-for-poly-games} studies the effect of action observability. Their empirical results show that access to other agents' actions, combined with belief-formation, can \textit{accelerate learning}.

Despite these advances, a key challenge in belief-based model-free learning for games is \textit{scalability}: as the number of agents increases, the joint action space grows exponentially. Thus, exploring (or estimating) the reward from every possible joint action profile could slow down convergence. 

A solution lies in the structure of the rewards. In this paper, we focus on \emph{anonymous games} \cite{Con_Games_with_Pl_Spec_Milchtaich, blonski_char-of-pure-strat-eq-in-finite-anon-games}, where an agent’s reward depends solely on the actions taken, and not on the identities of the agents taking them. This anonymity enables agents to focus on the aggregate distribution of actions, instead of the full combinatorial explosion of the individual moves. While supporting such simplification, these games are also widely used to model real-world systems, such as congestion in roadway networks \cite{Con_Games_with_Pl_Spec_Milchtaich}, interference in wireless networks \cite{han-et-al_game-theory-in-wireless-and-comm-nets}, and financial markets \cite{ohara_market-microstruct-theory}.

Building on the structure of anonymous games, we introduce \emph{aggregate fictitious play} (agg-FP). Instead of tracking each individual agent's actions, agents using agg-FP track the frequency of the \textit{number of others} playing each action, followed by a best-response to these aggregated beliefs. Overall, our main contributions are as follows:
\begin{itemize}
\item We study the problem of learning Nash equilibria in anonymous games through repeated play. Under the assumption that the game is known, we introduce agg-FP and show that it converges under the same conditions as FP for anonymous polymatrix games \cite{howson_eq-of-poly-games,cai-et-al_zero-sum-poly-generalization-of-minmax}.
\item We extend agg-FP to model-free settings and games with random payoffs by incorporating a reward estimation process that operates on a faster timescale (as in \cite{Conv_Multi_Timescale_RL_Algos_in_NFGs_Leslie_Collins}). For anonymous polymatrix games, we prove that the resulting dynamics converge to a neighborhood of a Nash equilibrium under the same conditions as FP.
\item Through simulations, we empirically demonstrate that two-timescale agg-FP converges faster in model-free settings than existing approaches (i.e. two-timescale FP and individual Q-learning).
\end{itemize}

We note that mean-field games \cite{lauriere-et-al_learning-in-mean-field-games-survey} and population games \cite{kara-martins_diff-eq-approx-for-population-games} address scalability in a similar manner: by tracking the distribution of agents over actions (or states). However, these models assume very large numbers of homogeneous agents. In contrast, anonymous games allow each agent to have a different reward function and support scalability in the mid-regime, where the number of agents is too small for mean-field approximations, yet too large for FP to converge quickly. Also, learning algorithms for mean-field games \cite{lauriere-et-al_learning-in-mean-field-games-survey} generally\footnote{\cite{zaman-et-al_oracle-free-learning-mean-field-games} presents an oracle-free algorithm for mean-field games, yet tracking the empirical frequencies of aggregate actions remains a simpler method.} require an oracle to compute the infinite-horizon distribution of agents. Our approach is more straightforward: we only track the empirical frequencies of aggregate actions.

\noindent\textbf{Preliminary Notation.} Let $\mathbb{N} := \{0,1,\dots\}$. Given a positive $M\in\mathbb{N}$ and a set $B$, we define $[M]:=\{1,\dots,M\}$ and $B^{\otimes M} := B\times\dots\times B$, the $M$-Cartesian product of $B$. Given an \( M \)-dimensional vector \( v \), we often use \( v^i \) to denote its \( i \)-th entry and \( v^{-i} \) to denote all entries except the $i$-th. We denote the indicator function by $\mathds{1}_{\{\cdot\}}$ and a vector of \textit{ones} by $\mathbf{1}$. With a slight abuse of notation, for $y$ in a finite set $\mathbb{Y}$,
% we write \(\mathds{1}\{y\}\) for the one-hot encoding of \(y\) when \(y\) belongs to a finite set \(\mathbb{Y}\). In particular, 
\(\mathds{1}\{y\}\) is a vector of zeros with a single \(1\) at the position corresponding to \(y\), when all the elements in $\mathbb{Y}$ are indexed. Finally, $X\sim p$ means that $X$ is a random variable with distribution $p$.

\section{Framework and Problem Description} \label{sec:Framework}

\subsection{$N$-Player Matrix Games}
 
An $N$-player matrix game (game for short) is characterized by a tuple $\mathcal{G}=(N, (\mathbb{A}^i)_{i=1}^N, (r^i)_{i=1}^N)$, where $N$ is the number of players, $\mathbb{A}^i$ is the finite action set of player $i$, and $r^i:\mathbb{A}^1\times\dots\times\mathbb{A}^N\to\mathbb{R}$ is the reward function of $i$. We let $\bar{\mathbb{A}}:= \mathbb{A}^1\times\dots\times \mathbb{A}^N$ and $r:=(r^i)_{i=1}^N$.

Each player \( i \) selects an action \( a^i \) from the set \( \mathbb{A}^i \) simultaneously, and receives a reward \( r^i(a^i, a^{-i}) \) ($a^{-i}$ is the actions of all players except \( i \); see preliminary notation in Section~\ref{sec:Intro}). Players can \textit{independently} randomize their actions by using (mixed) strategies. A strategy $\pi^i$ of player $i$ is a probability distribution over $\mathbb{A}^i$, where $\pi^i(a^i)$ is the probability of player $i$ selecting $a^i$. The expected reward of player $i$ under the strategy profile $\pi=(\pi^1,\dots,\pi^N)$ is 
$$R^i(\pi^i,\pi^{-i}):=\sum_{a\in \bar{\mathbb{A}}} r^i(a^i,a^{-i})\prod_{j=1}^N \pi^j(a^j).$$

The set of $\epsilon$-Nash equilibria $\mathbb{NE}_\epsilon(\mathcal{G})$ of $\mathcal{G}$ is the $\pi_*$ that satisfy the following set of inequalities for each player $i\in[N]$ and any $\pi^i\in\Delta(\mathbb{A}^i)$:
$$R^i(\pi^i_*, \pi^{-i}_*) \geq R^i(\pi^i, \pi^{-i}_*)-\epsilon,$$
where $\Delta(\mathbb{A}^i)$ is the set of all probability distributions on $\mathbb{A}^i$. If the above holds with $\epsilon = 0$, then $\pi_*$ is a Nash equilibrium: no player can improve its expected reward by unilaterally changing its strategy. We use $\mathbb{NE}(\mathcal{G})$ to denote the set of Nash equilibria of $\mathcal{G}$, which is always nonempty \cite{basar-olsder_dyn-noncoop-game-theory}.

A class of games that is of special interest to this paper is \textit{polymatrix} (or \emph{separable network}) games \cite{howson_eq-of-poly-games,cai-et-al_zero-sum-poly-generalization-of-minmax}, in which players' rewards can be expressed as a sum of their ``pairwise interactions.'' Formally, a game is polymatrix if for each $i$, there exist $r^{ij}:\mathbb{A}^i\times \mathbb{A}^j\to\mathbb{R}$, $j\in[N]\setminus\{i\}$, such that
\begin{align*}
    r^i(a^i,a^{-i}) = \sum_{j=1,~j\neq i}^N r^{ij}(a^i,a^j),\quad \text{ for all }a\in\bar{\mathbb{A}}.
\end{align*}

\subsection{Anonymous Games}

In many real-world multi-agent interactions, agents' actions impact each other in a uniform manner\footnote{In the paper, we use ``player'' and ``agent'' interchangeably.}. An example is financial markets \cite{ohara_market-microstruct-theory}, where an investor’s trading decision influences asset prices without targeting specific traders. Another example is network routing \cite{han-et-al_game-theory-in-wireless-and-comm-nets}, where the load on a communication network depends on the collective traffic from all users rather than any particular sender-receiver pair.

This concept is formalized in game theory as anonymous games, which is the primary focus of this paper. Specifically, in these games an agent's reward depends on other agents' actions only through which actions are taken, and not on which agent is playing which action \cite{Con_Games_with_Pl_Spec_Milchtaich,blonski_char-of-pure-strat-eq-in-finite-anon-games,Notions_of_anonymity_Ham,Sym_Plan}. To give a formal definition, we use $\rho$ to denote a permutation of $N-1$ elements.

\begin{definition}\label{def:anon_games}
A game is said to be \textit{anonymous} if $\mathbb{A}^1=\mathbb{A}^2=\dots=\mathbb{A}^N$ and $r^i(a^i, a^{-i}) = r^i(a^i, \rho(a^{-i}))$ for every agent $i$, any action profile $a\in\bar{\mathbb{A}}$, and any permutation $\rho$.
\end{definition}

We let $\mathbb{A}:=\mathbb{A}^1=\dots=\mathbb{A}^N$ and $n:=|\mathbb{A}|$.

\subsection{Games with Random Payoffs}

Random-payoff games generalize matrix games by allowing the rewards to be influenced by an exogenous random variable \cite{charnes-et-al_zero-zero-chance-cons-games,harsanyi_games-w-randomly-dist-payoffs}. Specifically, the reward function of each agent \( i \) depends on a random variable \( \Theta^i \). We denote by \( r^i_{\theta^i}(a^i, a^{-i}) \) the reward that agent \( i \) receives under the action profile \( a \) when the random variable \( \Theta^i \) has realization \( \theta^i \). We assume that each $\Theta^i$ has distribution $\mathcal{P}^i$, takes values in a bounded set $\Omega^i$, has finite mean and variance, and is independent across agents and their actions. 

We define $\Omega:=\Omega^1\times\dots\times\Omega^N$ and $\mathcal{P}:=\mathcal{P}^1\times\dots\times\mathcal{P}^N$, and denote a random-payoff game by $\mathcal{G}_{\mathcal{P}} := (N,(\mathbb{A}^i)_{i=1}^N,(r^i_{\Theta^i})_{i=1}^N,(\Omega,\mathcal{P}))$. We also write $r_\Theta := (r^i_{\Theta^i})_{i=1}^N$. The game $\mathcal{G}_{\mathcal{P}}$ is anonymous or polymatrix if $(N,(\mathbb{A}^i)_{i=1}^N,(r^i_{\theta^i})_{i=1}^N)$ is anonymous or polymatrix, respectively, for any realization $\theta:=(\theta^1,\dots,\theta^N)$ of $\Theta:=(\Theta^1,\dots,\Theta^N)$.

A typical objective is for each agent to maximize its expected reward \cite{harsanyi_games-w-randomly-dist-payoffs}. Accordingly, the set of $\epsilon$-Nash equilibria $\mathbb{NE}_\epsilon(\mathcal{G}_{\mathcal{P}})$ of $\mathcal{G}_{\mathcal{P}}$ is the $\pi_*$ that satisfy the following set of inequalities for each agent $i\in[N]$ and any $\pi^i\in\Delta(\mathbb{A}^i)$:
\begin{align*}
\mathbb{E}[R^i_{\Theta^i}(\pi^i_*,\pi^{-i}_*)] \geq \mathbb{E}[R^i_{\Theta^i}(\pi^i,\pi^{-i}_*)]- \epsilon,
\end{align*}
where the expectation is taken over $\Theta^i$.

\subsection{Learning in Repeated Games and {Problem Statement}}

We study the problem of learning a Nash equilibrium of an anonymous game, possibly with random-payoffs, through repeated play. In this setting---known as \textit{learning in (repeated) games} \cite{fudenberg_levine_th-of-learning-in-games}---agents engage through a fixed game over successive \textit{stages}, where at each stage (or timestep) they select actions, receive rewards, and update their strategies based on past experiences. We denote the action taken by agent \( i \) at stage \( k \) by \( a^i_k \), and assume that each agent can observe other agents' actions, but only receive information about its own reward.

When the stage game has random payoffs, we assume that the randomness at each stage \( k \) is independently and identically distributed as \( \Theta_k \sim \mathcal{P} \). As a result, the stage reward at each timestep $k$ is given by $r_{\Theta_k}$, where $\Theta_k\sim\mathcal{P}$.

We consider two scenarios: (i) The stage game is an anonymous matrix game and each agent $i$ knows $r^i$; (ii) The stage game is a random-payoff anonymous game, but neither the reward functions nor the distribution $\mathcal{P}$ are known to the agents. Our results from the model-based case in (i) will be the precursor for our results in the model-free setting in (ii).

\textit{Problem Statement:} We start by formalizing that, in anonymous games, the dependence of an agent's reward on other agents' actions can be ``succinctly represented'' as a function of their aggregate actions. Our goal is to develop a learning dynamics that exploits this structure to speed up learning an $\epsilon$-Nash equilibrium in repeated play.

\tikzstyle{agent} = [rectangle, rounded corners, minimum width=2cm, minimum height=1cm, text centered, draw=blue, fill=blue!10]
\tikzstyle{environment} = [rectangle, rounded corners, minimum width=3cm, minimum height=1.5cm, text centered, draw=black, fill=gray!10, align=center]
\tikzstyle{arrow} = [thick,->,>=stealth]
\begin{center}
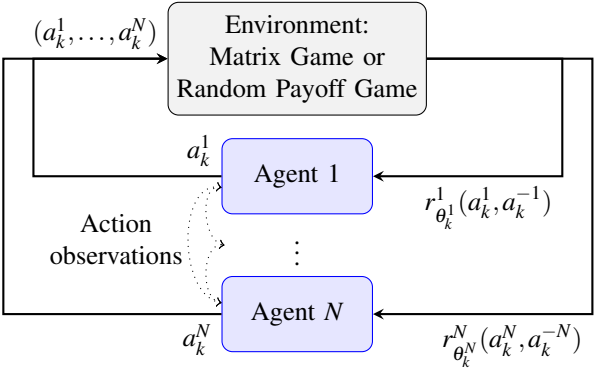
\begin{figure}
\begin{tikzpicture}

    % Environment
    \node (env) [environment] {Environment: \\ Matrix Game or\\ Random Payoff Game};

    % Agents (Stacked below the environment)
    \node (agent1) [agent, below=0.3cm of env] {Agent 1};
    % \node (agent2) [agent, below=0.2cm of agent1] {Agent 2};
    \node (dots) [below=0.02cm of agent1] {\vdots};
    \node (agentN) [agent, below=0.02cm of dots] {Agent $N$};

    % Actions from agents to environment
    \draw [arrow] (agent1.west) node[above left] {$a^1_k$} -- ++(-2.5,0) |-  (env.west);
    % \draw [arrow] (agent2.west) node[above left] {$a_2$} -- ++(-1.8,0) |- (env.west);
    \draw [arrow] (agentN.west) node[below left] {$a^N_k$} -- ++(-2.9,0) |- (env.west) node[above left] {$(a^1_k,\dots,a^N_k)$} ;

    % Feedback from environment to agents
    \draw [arrow] (env.east) -- ++(1.8,0) |- node[below left] {$r^1_{\theta^1_k}(a^1_k,a^{-1}_k)$} (agent1.east);
    % \draw [arrow] (env.east) -- ++(1.3,0) |- node[below left] {$r_2(a^2_k,a^{-2}_k)$} (agent2.east);
    \draw [arrow] (env.east) -- ++(2.2,0) |- node[below left] {$r^N_{\theta^N_k}(a^N_k,a^{-N}_k)$} (agentN.east);

    \draw [dotted, <->] ([yshift=-0.15cm]agent1.west) to[out=170, in=170] ([xshift=-0.8cm]dots.west);
    \draw [dotted, <->] ([yshift=0.15cm]agentN.west) to[out=170, in=170] ([xshift=-0.8cm]dots.west);
    \draw [dotted, <->] ([yshift=-0.15cm]agent1.west) to[out=170, in=170] node[xshift=0.05cm, left, align=center] {Action\\observations} ([yshift=0.15cm]agentN.west);

\end{tikzpicture}
\caption{Illustration of learning in repeated games.}
\label{fig:env}
\end{figure}
\end{center}

\section{Succinct Representation and Anonymity} \label{sec:Anon_Games}

Let us define $\sigma:\mathbb{A}^{\otimes(N-1)}\to \mathbb{R}^n$ as 
$$\sigma(a^{-i}) := \sum_{j=1,~j\neq i}^N \mathds{1}\{a^j\}.$$ To simplify notation, we will also use $x^{-i} := \sigma(a^{-i})$. Thus, $x^{-i}$ is the vector containing the \textit{counts of agents, excluding agent $i$, playing each action}, taking values in $\mathbb{X}:=\{\xi\in\mathbb{N}^n ~:~ \sum_{i=1}^n\xi_i=N-1 \}.$

Several papers define anonymous games using Definition~\ref{def:anon_games} \cite{Sym_Plan,Notions_of_anonymity_Ham}, whereas others define them by requiring the reward functions to have the form $r^i(a^i, x^{-i})$ \cite{Con_Games_with_Pl_Spec_Milchtaich,Comp_Eq_in_Anon_Games_Daskalakis_Papadimitrou,Symm_and_the_Comp_of_Pure_NE_Brandt_et_al}. Although both definitions are widely used, it appears that a complete proof of their equivalence is not available. While \cite{Symm_and_the_Comp_of_Pure_NE_Brandt_et_al} offers a justification for this equivalence, it does not provide a comprehensive proof. We formalize their explanation below.

\begin{lemma}\label{lem:agg_rep}
A game $\mathcal{G}$ is anonymous if and only if, for each agent $i\in[N]$, there is a function $\bar{r}^i:\mathbb{A}\times \mathbb{X}\to\mathbb{R}$ such that $r^i(a^i, a^{-i}) = \bar{r}^i(a^i, x^{-i})$  for all $a\in\mathbb{A}^{\otimes N}$.
\end{lemma}

\begin{proof}
The if direction is straightforward: a reward function of the form $\bar{r}^i(a^i, x^{-i})$ implies that permuting other agents' actions does not affect the reward, as it leaves $x^{-i}$ unchanged.

We proceed with the only if direction. Observe that the set of $(N-1)$-permutations induces an equivalence relation $\approx$ on $\mathbb{A}^{\otimes(N-1)}$. That is, we set $a^{-i}\in\mathbb{A}^{\otimes(N-1)}$ to be equivalent to $b^{-i}\in\mathbb{A}^{\otimes(N-1)}$ if $a^{-i}$ is a permutation of $b^{-i}$. Let us denote the corresponding quotient space by $\mathbb{A}^{\otimes(N-1)}/\approx $.

Since the game is anonymous, we have $r^i(a^i, a^{-i}) = r^i(a^i, \rho(a^{-i}))$ for every $a\in\mathbb{A}^{\otimes N}$ and permutation $\rho$. Thus, $r^i(a^i, \cdot)$ is constant within each equivalence class of $\approx $. This results in the function $\tilde{r}^i$ given by $\tilde{r}^i(a^i, [a^{-i}]_{\approx }) := r^i(a^i, a^{-i})$ to be well-defined for all $a^i\in\mathbb{A}$ and  $[a^{-i}]_{\approx}\in\mathbb{A}^{\otimes(N-1)}/\approx$.

Now, for any \(a^{-i}\) and \(b^{-i}\) in \(\mathbb{A}^{\otimes(N-1)}\), let \(x^{-i} = \sigma(a^{-i})\) and \(y^{-i} = \sigma(b^{-i})\). It is clear that if \(a^{-i}\) is a permutation of \(b^{-i}\), then \(x^{-i} = y^{-i}\). Conversely, assume that \(x^{-i} = y^{-i}\). We claim that \(a^{-i}\) must then be a permutation of \(b^{-i}\). To see this, suppose, by contradiction, that \(a^{-i}\) is not a permutation of \(b^{-i}\). Remove one by one every action that appears in both \(a^{-i}\) and \(b^{-i}\), until no common actions remain; denote the remaining lists by \(\tilde{a}^{-i}\) and \(\tilde{b}^{-i}\). Since \(a^{-i}\) is not a permutation of \(b^{-i}\), some actions remain unmatched, and thus both \(\tilde{a}^{-i}\) and \(\tilde{b}^{-i}\) are nonempty and have no common elements (otherwise those elements would have been  removed). Therefore,  $\sum_{l=1}^{|\tilde{a}^{-i}|}\mathds{1}\{\tilde{a}^{l}\} \neq \sum_{l=1}^{|\tilde{b}^{-i}|}\mathds{1}\{\tilde{b}^{l}\}$. Since the removed common actions contribute equally to both $x^{-i}$ and $y^{-i}$, this inequality implies \(x^{-i} \neq y^{-i}\), contradicting our assumption. Hence, \(a^{-i}\) must be a permutation of \(b^{-i}\). 

As a result, $x^{-i} = y^{-i}$ if and only if $a^{-i}$ is a permutation of $b^{-i}$, which establishes a bijection between $\mathbb{A}^{\otimes(N-1)}/\approx $ and $\mathbb{X}$. Thus, we can characterize each element $[a^{-i}]_{\approx }$ of $\mathbb{A}^{\otimes(N-1)}/\approx $ with a \textit{single element} $x^{-i}\in \mathbb{X}$, allowing us to define a function $\bar{r}^i$ on $\mathbb{A}\times \mathbb{X}$ as
\begin{align*}
\bar{r}^i(a^i, x^{-i}) := \tilde{r}^i(a^i, [a^{-i}]_{\approx }) = r(a^i, a^{-i})
\end{align*}
for any $a^i\in\mathbb{A}$ and $x^{-i}\in \mathbb{X}$.
\end{proof}

\begin{remark}\label{rem:complexity} 
Importantly, $\bar{r}^i$ is defined on $\mathbb{A} \times \mathbb{X}$, which can assume a total of  
$$
|\mathbb{A}\times \mathbb{X}| = n {N+n-2 \choose n-1}
$$  
distinct values, compared to the $n^N$ values required for $r^i$. For example, with 3 actions and 5 agents, representing the game using $\bar{r}^i$ reduces the number of distinct reward entries from 243 to 45. \qedwhite
\end{remark}

\begin{remark}
Games where rewards can be encoded compactly, such as anonymous games, are referred to as being succinctly representable. We refer to $\bar{r}^i$ as the \textit{succinct representation} of $r^i$. \qedwhite
\end{remark}

Now, consider the situation where each agent $i$ selects its action using a strategy $\pi^i$. Hence, we have a stochastic action profile $A:=(A^1,\dots,A^N)$, taking values in $\mathbb{A}^{\otimes N}$, where $A^i\sim\pi^i$. Let $X^{-i}:=\sigma(A^{-i})$. We denote the distribution of $X^{-i}$ (resulting from $\pi^{-i}$) by $\mu^i$:
\begin{align*}
\mu^i(x^{-i}) := \mathbb{P}(X^{-i}=x^{-i}),\quad \text{ for all }x^{-i}\in \mathbb{X}.
\end{align*}
The representation in Lemma~\ref{lem:agg_rep} extends to mixed strategies under $(\mu^1,\dots,\mu^N)$.

\begin{lemma}\label{lem:agg_rep_mixed}
The succinct representation $\bar{r}^i$ of $r^i$ satisfies
\begin{align*}
R^i(\pi^i, \pi^{-i}) = \bar{R}^i(\pi^i, \mu^i), \quad\text{for all }\pi\in (\Delta(\mathbb{A}))^N.
\end{align*}
\end{lemma}

\begin{proof}
For any $\pi\in(\Delta(\mathbb{A}))^N$, we have
\begin{align*}
& \bar{R}^i(\pi^i, \mu^i)\\
&= \sum_{a^i\in\mathbb{A}}\sum_{x^{-i}\in \mathbb{X}} \bar{r}^i(a^i, x^{-i})\pi^i(a^i)\mu^i(x^{-i})\\
&= \sum_{a^i\in\mathbb{A}}\sum_{x^{-i}\in \mathbb{X}}\sum_{\{a^{-i}~:~\sigma(a^{-i})=x^{-i}\}} r^i(a^i, a^{-i})\pi^i(a^i)\mathbb{P}(A^{-i}=a^{-i}).
\end{align*}
The combined summation over \( x^{-i} \in  \mathbb{X} \) and \( \{a^{-i} ~:~ \sigma(a^{-i}) = x^{-i} \} \) is equivalent to summing over \( a^{-i} \) across the entire set $\mathbb{A}^{\otimes(N-1)}$. Thus, $\bar{R}^i(\pi^i, \mu^i) = R^i(\pi^i, \pi^{-i})$.
\end{proof}

\section{Aggregate Fictitious Play} \label{sec:Rep_Anon_Games_and_Agg_FP}

We have so far formalized the succinct representation of anonymous games. Now, we turn to the problem of learning their Nash equilibria through repeated interactions---a setting we briefly introduced in Section~\ref{sec:Framework}. Specifically, we assume that $\mathcal{G}$ is an anonymous game, played repeatedly, and that each agent $i$ knows $r^i$. In this context, we will introduce aggregate fictitious play (agg-FP): a new learning dynamics that modifies traditional FP by leveraging anonymity.

\subsection{Aggregate Fictitious Play in Repeated Games} \label{subsec:Agg_FP}

\subsubsection{Background on FP}

Introduced by \cite{brown_some-notes-on-comp-of-games-sols}, FP is a learning process in which each agent estimates other agents' strategies using the weighted average of past actions and play best responses to these estimates. In particular, at each timestep $k\geq 1$, agent $i$ updates its estimate of agent $j$'s strategy using
\begin{align}
\hat{\pi}^j_{k} = \hat{\pi}^j_{k-1} + \alpha_k (\mathds{1}\{a^j_{k}\} - \hat{\pi}^j_{k-1}), \label{eq:pi_hat} 
\end{align}
where $a^j_k$ is the action taken by $j$ at time $k$, and $(\alpha_k)_{k\geq 1}$ is a step-size sequence that satisfies the Robbins-Monro
conditions
\begin{align}
\sum_{k\geq 1}\alpha_k = \infty,\quad \sum_{k\geq 1}\alpha_k^2 < \infty. \label{eq:stepsize_conds}
\end{align}
Assuming that other agents are playing $\hat{\pi}^{-i}_k$ as stationary strategies, agent $i$ then plays a best response:  
\begin{align*}
a^i_{k+1} \in \arg\max_{a^i \in \mathbb{A}} \{R^i(a^i, \hat{\pi}^{-i}_k)\}.
\end{align*}
The resulting $(\hat{\pi}_k)_{k\geq 0}$, which is commonly referred to as the \textit{belief} \cite{fudenberg_levine_th-of-learning-in-games,ozdaglar-et-al_independent_learning-stochastic-games}, has been shown to converge to $\mathbb{NE}(\mathcal{G})$ for zero-sum polymatrix games \cite{robinson_an-iter-method-for-solving-a-game, Fict_play_in_nets_Ewerhart_Valkanova} and potential games \cite{monderer_shapley_pot-games} (see \cite{ozdaglar-et-al_independent_learning-stochastic-games} for additional types of games where FP converges).

\begin{remark} \label{rem:FP_init}
Assume arbitrary initial actions $a_0$ and define $\hat{\pi}_0 = \mathds{1}\{a_0\}$. Then, Eq. \eqref{eq:pi_hat} specifies $\hat{\pi}_k$ for all $k\geq 1$. \qedwhite
\end{remark}

\begin{remark}
Throughout the paper, we adopt an arbitrary deterministic tie-breaking rule to select an action when $\argmax$ yields multiple maximizers. One such rule is selecting the maximizer with the smallest index. \qedwhite
\end{remark}

\subsubsection{Aggregate Fictitious Play}

We can interpret $\hat{\pi}^{-i}_k$ as how agent $i$ models other agents at time $k$, which constitutes modeling each of them individually. 

Building on the ideas in Section~\ref{sec:Anon_Games}, we propose an aggregative approach to modeling other agents' behavior. For any given action history $(a_k)_{k\geq 0}$, we define $(\hat{\mu}_k)_{k\geq 0}$ as
\begin{align}
&\hat{\mu}^i_{k} = \hat{\mu}^i_{k-1} + \alpha_{k}(\mathds{1}\{x^{-i}_{k}\}-\hat{\mu}^i_{k-1}) \label{eq:mu_hat}
\end{align}
for each agent $i\in[N]$ and stage $k\geq 1$, where $x^{-i}_k:=\sigma(a^{-i}_k)$. Hence, $(\hat{\mu}_k)_{k \geq 0}$ is the weighted empirical frequencies of \textit{aggregate actions}. Accordingly, we refer to $(\hat{\mu}_k)_{k \geq 0}$ as the \textit{aggregate belief}.

\begin{remark}
In line with Remark~\ref{rem:FP_init}, we set $\hat{\mu}^i_0 = \mathds{1}\{x^{-i}_0\}$ for all $i$. Then, Eq. \eqref{eq:mu_hat} specifies $\hat{\mu}_k$ for all $k\geq 1.$ \qedwhite
\end{remark}

\begin{remark}\label{rem:different_beliefs}
Both $\hat{\mu}^i_k$ and $\hat{\pi}^{-i}_k$ define probability distributions on $\mathbb{X}$, which can be used as beliefs for $x^{-i}_{k+1}$. Nonetheless, these distributions are \textit{not the same in general}. In other words, assuming $A^j\sim\hat{\pi}^j_{k}$ for all $j\in[N]$, there may exist $(a_l)_{0\leq l\leq k}$ and $x^{-i}\in\mathbb{X}$ such that
\begin{align}
\mathbb{P}(\sigma(A^{-i}) = x^{-i}) \neq \hat{\mu}^i_{k}(x^{-i}). \label{eq:uneq_agg_op_strats}
\end{align}
Therefore, $\hat{\mu}_k$ and $\hat{\pi}_k$ can yield different beliefs. \qedwhite
\end{remark}

Remark~\ref{rem:different_beliefs} shows that, in general, the aggregative and individualistic approaches produce different beliefs. Yet, the following lemma reveals an important connection: they yield the \textit{same expected rewards in polymatrix games}.

\begin{lemma}\label{lem:equal_exp_rews}
Assume that $\mathcal{G}$ is anonymous and polymatrix. Then, for each stage $k\geq 0$, agent $i\in[N]$, and any history of joint actions $(a_l)_{0\leq l\leq k}$, it holds that
\begin{align}
\bar{R}^i(\cdot,\hat{\mu}^{i}_k) = R^i(\cdot,\hat{\pi}^{-i}_k). \label{eq:equal_exp_payoffs}
\end{align}
\end{lemma}

\begin{proof}
Note that \eqref{eq:equal_exp_payoffs} holds trivially for $k=0$. As for $k\geq 1$, for any $i,j\in[N]$, we have $\hat{\pi}^j_k = \sum_{l=0}^k\bar{\alpha}_l \mathds{1}\{a_l^j\}$ and $\hat{\mu}^i_k = \sum_{l=0}^k\bar{\alpha}_l \mathds{1}\{x^{-i}_l\}$, where $\bar{\alpha}_k := \alpha_k$ and $\bar{\alpha}_l := \alpha_k\prod_{m=l+1}^{k}(1-\alpha_m)$ for all $l\in\{0,\dots,k-1\}$. As a result,
\begin{align*}
R^i(a^i,\hat{\pi}^{-i}_k) &= \sum_{a^{-i}\in\mathbb{A}^{\otimes(N-1)}} \sum_{j\neq i} r^i(a^i,a^j)\sum_{l=0}^k\bar{\alpha}_l \mathds{1}_{\{a^j=a_l^j\}}\\
&= \sum_{l=0}^k \bar{\alpha}_l \sum_{a^{-i}\in\mathbb{A}^{\otimes(N-1)}} \sum_{j\neq i} r^i(a^i,a^j) \mathds{1}_{\{a^j=a_l^j\}}\\
&= \sum_{l=0}^k \bar{\alpha}_l \sum_{x^{-i}\in\mathbb{X}} \bar{r}^i(a^i,x^{-i}) \mathds{1}_{\{x^{-i}=\sigma(a_l^{-i})\}}\\
&= \sum_{x^{-i}\in\mathbb{X}} \bar{r}^i(a^i,x^{-i}) \sum_{l=0}^k \bar{\alpha}_l \mathds{1}_{\{x^{-i}=\sigma(a_l^{-i})\}}\\
&= \bar{R}^i(a^i,\hat{\mu}^i_k).
\end{align*}
\end{proof}

\begin{remark}We emphasize that the equivalence of expected rewards, stated in Lemma~\ref{lem:equal_exp_rews}, is specific to polymatrix games. The polymatrix structure of the game allows the expected reward to be written as a double sum---first over time indices \( l = 0, \dots, k \), and second over agents \( j \in \{1, \dots, N\} \setminus \{i\} \). This structure enables us to swap the order of summation and then collect the terms over agents, resulting in an aggregative representation. In contrast, for non-polymatrix games, this reordering is generally invalid, leading to different expected rewards under individual versus aggregate beliefs.\qedwhite\end{remark}

A direct yet important consequence of Lemma~\ref{lem:equal_exp_rews} is the corollary below.

\begin{corollary}\label{cor:same_actions}
If $\mathcal{G}$ is polymatrix and anonymous, then $\bar{R}^i(\cdot,\hat{\mu}^{i}_k)$ and $R^i(\cdot,\hat{\pi}^{-i}_k)$ are best-response equivalent ($i$'s best response is the same regardless of which reward vector it uses). More generally, $(\hat{\mu}_l)_{0\leq l\leq k}$ and $(\hat{\pi}_l)_{0\leq l\leq k}$ generate the same $a_{k+1}$ under any deterministic decision mechanism that maps expected rewards to actions.
\end{corollary}

Recall that FP combines $(\hat{\pi}_k)_{k\geq 0}$ with best-response action-selection. Inspired by this, we simply define aggregate fictitious play as best responding to aggregate beliefs.

\begin{definition}
In \textit{aggregate fictitious play (agg-FP)}, at each time $k\geq 1$, each agent $i$ updates its aggregate belief and selects an action according to
\begin{subequations} \label{eq:agg_fp}
\begin{align}
&\hat{\mu}^i_{k} = \hat{\mu}^i_{k-1} + \alpha_{k}(\mathds{1}\{x^{-i}_{k}\}-\hat{\mu}^i_{k-1}), \label{eq:mu_with_br}\\
& a^i_{k+1} \in\argmax_{a^i\in\mathbb{A}} \{\bar{R}^i(a^i,\hat{\mu}^i_k)\}. \label{eq:agg_br}
\end{align}
\end{subequations}
% We remind that, in cases where \eqref{eq:agg_br} has multiple maximizers, an arbitrary deterministic tie-breaking rule determines which one is played.
% % and updates $\hat{\mu}^i_{k+1}$ according to \eqref{eq:mu} using the played actions.
\end{definition}

From Corollary~\ref{cor:same_actions}, for polymatrix anonymous games, it follows that FP and agg-FP produce identical action trajectories. As a result, in such games, agg-FP inherits the convergence properties of FP. However, before formally stating this, we must first clarify what we mean by the convergence of FP and extend this notion to agg-FP. Convergence of FP 
% does not mean that the agents' actions necessarily converge to a Nash equilibrium (a pure-strategy Nash equilibrium might not even exist). Instead, it 
refers to the convergence of the beliefs $(\hat{\pi}_k)_{k \geq 0}$ to $\mathbb{NE}(\mathcal{G})$ \cite{fudenberg_levine_th-of-learning-in-games, ozdaglar-et-al_independent_learning-stochastic-games}. Akin to this, we say that agg-FP converges if the resulting empirical action frequencies converge to $\mathbb{NE}(\mathcal{G})$.

\begin{definition}
Let $(a_k)_{k \geq 0}$ be a joint action sequence generated by agg-FP. Initialize $\hat{\gamma}_0 = \mathds{1}\{a_0\}$ and define the empirical action frequencies $(\hat{\gamma}_k)_{k \geq 0}$ for all $i \in [N]$ and $k \geq 1$ as
\begin{align}
& \hat{\gamma}^i_{k} = \hat{\gamma}^i_{k-1} + \alpha_k(\mathds{1}\{a^i_{k}\} - \hat{\gamma}^i_{k-1}). \label{eq:agg_fp_actions}
\end{align}
We say that agg-FP converges if
% $(\hat{\gamma}_k)_{k\geq 0}$ converges to $\mathbb{NE}(\mathcal{G})$, meaning that
\begin{align*}
\lim_{k\to\infty}\inf_{\pi_*\in\mathbb{NE}(\mathcal{G})}\|\hat{\gamma}_k - \pi_*\| = 0.
\end{align*}
\end{definition}

The lemma below connects the convergence properties of agg-FP and FP, which we state without proof, as it follows directly from Corollary~\ref{cor:same_actions}.

\begin{lemma}\label{lem:aggFP_conv}
In anonymous polymatrix games, agg-FP converges (to a Nash equilibrium) if and only if FP does.
\end{lemma}

\begin{corollary}
Since FP converges in polymatrix zero-sum \cite{Fict_play_in_nets_Ewerhart_Valkanova} and potential games \cite{monderer_shapley_pot-games}, so does agg-FP.
\end{corollary}

\begin{remark}
Notice that Lemma~\ref{lem:aggFP_conv} holds for polymatrix games only. In \textit{non-polymatrix} anonymous games, agg-FP and FP can generate different action trajectories. This calls for further research into whether agg-FP can converge in non-polymatrix anonymous games where FP does not. \qedwhite
\end{remark}

\begin{remark} \label{rem:exploration}
In the upcoming analysis, we focus on the $\delta$-greedy variant of agg-FP. In this version, with probability $\delta \in (0,1)$, all agents ``collectively explore'' by selecting actions uniformly at random, while with probability $(1 - \delta)$ they ``act greedily'' by playing their best responses. Thus, $x^{-i}_k$ in \eqref{eq:mu_with_br} becomes a random variable that takes value $\sum_{j\neq i}\mathds{1}\{\argmax_{a^j\in\mathbb{A}}\{\bar{R}^j(a^j,\hat{\mu}_k^j)\}\}$ with probability $\delta/|\mathbb{X}|+1-\delta$, and any other value $\xi^{-i}$ in $\mathbb{X}$ with probability $\delta/|\mathbb{X}|$. This mechanism ensures sufficient exploration in the model-free setting, which we cover in Section~\ref{sec:Model_Free}. 

We note that the exploration event is dependent across agents; that is, at each iteration, either all agents explore, or none do. Nonetheless, we believe that our results can be extended, without much difficulty, to settings with agent-wise independent exploration, time-dependent exploration probabilities (e.g., decaying over time), and agent-specific exploration rates.
\qedwhite
\end{remark}

\subsection{Differential Equation Approximations}

\subsubsection{Aggregate Best-Response Dynamics}

Best-response (BR) dynamics is regarded as the continuous-time counterpart of FP \cite{fudenberg_levine_th-of-learning-in-games}. For all $\delta \in (0,1)$, $i \in [N]$, and $t \geq 0$, the $\delta$-greedy version of this dynamics is defined as
\begin{align}
\dot{\pi}^i_t = (1-\delta)\mathds{1}\left\{\argmax_{a^i\in\mathbb{A}}\{R^i(a^i,\pi^{-i}_t)\}\right\} + \frac{\delta\mathbf{1}}{n} - \pi^i_t. \label{eq:brd}
\end{align}
We will avoid restating that \eqref{eq:brd} represents the $\delta$-greedy version, and will often refer to it as the BR dynamics. In \cite{benaim_et_al_stoch-approx-and-diff-inclus}, the authors show that the limit sets of $(\hat{\pi}_k)_{k\geq 0}$ and $(\pi_t)_{t\geq 0}$ coincide whenever the BR dynamics has globally asymptotically stable equilibria, providing an alternative way to analyze the long-term behavior of FP.

We derive the continuous-time counterpart of $\delta$-greedy agg-FP based on the approach in \cite{benaim_et_al_stoch-approx-and-diff-inclus}. Recalling the exploration scheme in Remark~\ref{rem:exploration} and letting $a^j_{k,*}:=\argmax_{a^j\in\mathbb{A}}\{\bar{R}^j(a^j,\hat{\mu}^j_k)\}$ for all $j\in[N]$, we can write
\begin{align}
&\hat{\mu}^i_{k+1} - \hat{\mu}^i_k - \alpha_{k}\mathcal{E}^i_k \nonumber\\
% &= \alpha_k\left((1-\delta) \mathds{1}\left\{\sum_{j\neq i}\argmax_{a^j\in\mathbb{A}}\{\bar{R}^j(a^j,\hat{\mu}^j_k)\}\right\}\right) + \frac{\delta \mathbf{1}}{|\mathbb{X}|} - \hat{\mu}^i_{k}, \label{eq:mu_for_stoch_approx}
&= \alpha_k\left((1-\delta) \mathds{1}\{\sigma(a^{-i}_{k,*})\} + \frac{\delta \mathbf{1}}{|\mathbb{X}|} - \hat{\mu}^i_{k}\right), \label{eq:mu_for_stoch_approx}
\end{align}
where $\mathcal{E}^i_k$ is the approximation error given by
\begin{align*}
% \mathcal{E}^i_k &:= \mathds{1}\{x^{-i}_k\} - (1-\delta)\mathds{1}\left\{ \sum_{j\neq i}\argmax_{a^j\in\mathbb{A}}\{\bar{R}^j(a^j,\hat{\mu}^j_k)\}\right\} - \frac{\delta\mathbf{1}}{|\mathbb{X}|}\\
\mathcal{E}^i_k &:= \mathds{1}\{x^{-i}_k\} - (1-\delta)\mathds{1}\{\sigma(a^{-i}_{k,*})\} - \frac{\delta\mathbf{1}}{|\mathbb{X}|}\\
&= \mathds{1}\{x^{-i}_k\} -\mathbb{E}[\mathds{1}\{x^{-i}_k\}~|~\hat{\mu}_k].
\end{align*}
Since $(\mathcal{E}^1_k,\dots,\mathcal{E}^N_k)_{k\geq 0}$ is a zero-mean martingale difference sequence and the reward values are bounded, \eqref{eq:mu_for_stoch_approx} satisfies \cite[Definition~III~\&~IV]{benaim_et_al_stoch-approx-and-diff-inclus}. In addition, $(\alpha_k)_{k \geq 0}$ satisfies \eqref{eq:stepsize_conds}. As a result, \cite{benaim_et_al_stoch-approx-and-diff-inclus} yields the following system, with $i \in [N]$ and $t \geq 0$:
\begin{align}
&\dot{\mu}^i_t = (1-\delta)\mathds{1}\left\{\sum_{j\neq i}\mathds{1}\left\{\argmax_{a^j\in\mathbb{A}}\{\bar{R}^j(a^j,\mu^{j}_t)\}\right\}\right\} + \frac{\delta\mathbf{1}}{|\mathbb{X}|} - \mu^i_t. \label{eq:aggbr_mu}
\end{align}
Similarly, the set of differential equations below, with $i\in[N]$, is the continuous-time approximation of \eqref{eq:agg_fp_actions}:
\begin{align}
&\dot{\gamma}^i_t = (1-\delta)\mathds{1}\left\{\argmax_{a^i\in\mathbb{A}}\{\bar{R}^i(a^i,\mu^{i}_t)\}\right\} + \frac{\delta\mathbf{1}}{n} - \gamma^i_t. \label{eq:aggbr_gamma}
\end{align}
We refer to \eqref{eq:aggbr_mu}-\eqref{eq:aggbr_gamma} jointly as the ($\delta$-greedy) \textit{aggregate best-response (agg-BR) dynamics}.

\begin{remark}
If the equilibria of \eqref{eq:aggbr_mu} admits a Lyapunov function, then it follows from \cite{benaim_et_al_stoch-approx-and-diff-inclus} that the limit sets of $(\hat{\mu}_k)_{k\geq 0}$ and $(\hat{\gamma}_k)_{k\geq 0}$ are the equilibria of \eqref{eq:aggbr_mu} and \eqref{eq:aggbr_gamma}, respectively. \qedwhite
\end{remark}

\subsubsection{Relating BR to Agg-BR}

We proceed by comparing the trajectories of BR and agg-BR dynamics (like we did for FP and agg-FP in Section~\ref{subsec:Agg_FP}). To do so, we first need them to have \textit{consistent initial values}, meaning that
\begin{subequations} \label{eq:init_vals}
\begin{align}  
&\gamma_0 = \pi_0, \label{eq:init_vals_gamma}\\
&\mu^i_0(x^{-i}) = \sum_{\{a^{-i}:\sigma(a^{-i}) = x^{-i}\}} \prod_{j \neq i} \pi^j_0(a^j),~\text{for all } x^{-i} \in \mathbb{X}. \label{eq:init_vals_mu}
\end{align}
\end{subequations}
Notice that \eqref{eq:init_vals_mu} is equivalent to $ \mu^i_0(x^{-i}) = \mathbb{P}(\sigma(A^{-i}) = x^{-i}) $, where $ A^j \sim \pi^j_0 $ for all $ j $.

\begin{remark}
The initial conditions of BR and agg-BR are consistent if they are matched with those of FP and agg-FP, that is, if $\pi_0 = \gamma_0 = \mathds{1}\{a_0\}$ and $\mu^i_0 = \mathds{1}\{\sigma(a^{-i}_0)\}$ for all $i$. \qedwhite
\end{remark}

As in Remark~\ref{rem:different_beliefs}, assuming consistent initial values and letting $A^j \sim \pi^j_t$ for all $j \in [N]$, we can verify that there exists $t > 0$ and $x^{-i}\in\mathbb{X}$ such that
\begin{align*}
\mathbb{P}(\sigma(A^{-i})=x^{-i}) &\neq \mu^i_t(x^{-i}).
\end{align*}
Thus, $(\pi^{-i}_t)_{t\geq 0}$ and $(\mu^i_t)_{t\geq 0}$ eventually become ``misaligned.'' Note that this is the case even when $\mathcal{G}$ is polymatrix.

Nevertheless, analogous to Lemma~\ref{lem:equal_exp_rews}, $(\mu_t)_{t\geq 0}$ and $(\pi_t)_{t\geq 0}$ generate \textit{identical reward trajectories} whenever $\mathcal{G}$ is polymatrix. Crucially, this also leads to $\gamma_t = \pi_t$ for all $t \geq 0$.

\begin{lemma} \label{lem:aggbr_br_rew_equiv}
If $\mathcal{G}$ is polymatrix and anonymous, and \eqref{eq:init_vals} holds, then $\pi^i_t = \gamma^i_t$ and $R^i(\cdot,\pi^{-i}_t) = \bar{R}^i(\cdot,\mu^i_t)$ for all $i\in[N]$ and $t\geq 0$.
\end{lemma}

\begin{proof}
Notice that \eqref{eq:init_vals_mu} implies 
\begin{align*}
R^i(\cdot,\pi^{-i}_0) &= \sum_{j\neq i} \sum_{a^{-i}\in\mathbb{A}^{\otimes(N-1)}} r^i(\cdot, a^{-i}) \prod_{j\neq i}\pi^{j}_0(a_j)\\
&= \sum_{j\neq i} \sum_{x^{-i}\in\mathbb{X}}\sum_{\{a^{-i}~:~\sigma(a^{-i}) = x^{-i}\}} \bar{r}^i(\cdot, x^{-i}) \prod_{j\neq i}\pi^{j}_0(a_j)\\
% &= \sum_{j\neq i} \sum_{x^{-i}\in\mathbb{X}} \bar{r}^i(\cdot, x^{-i}) \mu^{i}_0(x^{-i})\\
&= \bar{R}^i(\cdot,\mu^{i}_0).
\end{align*}

Therefore, it is sufficient to show that the equality $dR^i(\cdot,\pi^{-i}_t)/dt = d\bar{R}^i(\cdot,\mu^{i}_t)/dt$ holds whenever $R^i(\cdot,\pi^{-i}_t) = \bar{R}^i(\cdot,\mu^{i}_t)$. Assuming that $R^i(a^i,\pi^{-i}_t) = \bar{R}^i(a^i,\mu^{i}_t)$ for all $a^i\in\mathbb{A}$ and letting $a^j_{t,*} := \argmax_{a^j\in\mathbb{A}}\{R^j(a^j,\pi^{-j}_t)\}$, we have
\begin{align*}
&\frac{dR^i(a^i,\pi^{-i}_t)}{dt}\\
&= (1-\delta)\sum_{a^{-i}\in\mathbb{A}^{\otimes(N-1)}}\sum_{j\neq i}r^i(a^i,a^j)\mathds{1}_{\{a^j=a^j_{t,*}\}}\\
&\quad +\frac{\delta}{n}\sum_{a^{-i}\in\mathbb{A}^{\otimes(N-1)}}\sum_{j\neq i}r^i(a^i,a^j) -\sum_{a^{-i}\in\mathbb{A}^{\otimes(N-1)}}\sum_{j\neq i}r^i(a^i,a^j)\pi^j_t(a^j)\\
&= (1-\delta)\sum_{x^{-i}\in\mathbb{X}} \bar{r}^i(a^i,x^{-i})\mathds{1}_{\left\{ x^{-i}=\sigma(a^{-i}_{t,*}) \right\}}\\
&\quad +\frac{\delta}{|\mathbb{X}|}\sum_{x^{-i}\in\mathbb{X}}\bar{r}^i(a^i,x^{-i}) - \bar{R}^i(a^i,\mu^i_t)\\
&= \frac{d\bar{R}^i(a^i,\mu^{i}_t)}{dt}.
\end{align*}
\end{proof}

\begin{corollary} \label{cor:same_conv_br_agg_br}
If $\mathcal{G}$ is polymatrix and anonymous, then the limit sets of $(\gamma_t)_{t\geq 0}$ and $(\pi_t)_{t\geq 0}$ are the same.
\end{corollary}

\section{Two-Timescale Agg-FP for Model-Free Learning} \label{sec:Model_Free} 

% To implement agg-FP, each agent $i$ needs to know \( r^i \), which we assumed to be an anonymous matrix game. We now extend agg-FP to a \textit{model-free} setting while also slightly generalizing the reward structure to accommodate anonymous games with random payoffs. Recalling the discussion in Section~\ref{sec:Framework}, this means that at each stage \( k \), the stage game is given by \( r_{\Theta_k} \), where \( \Theta_k \) are independently sampled from $\mathcal{P}$. We denote a generic (independent) copy of $\Theta_k$ by $\Theta$.

% To account for the uncertainty inherent in model-free settings, we denote agent \( i \)'s action at stage \( k \) by \( A^i_k \) and its corresponding reward by $R^i_k = r^i_{\Theta^i_k}(A^i_k, A^{-i}_k).$ Future actions are determined solely based on the history of past actions and individually observed rewards.

To implement agg-FP, each agent $i$ needs to know \( r^i \), which we assumed to be an anonymous matrix game. We now extend agg-FP to \textit{model-free} settings while also slightly generalizing the reward structure to accommodate anonymous games with random payoffs. 

As outlined in Section~\ref{sec:Framework}, games with random payoffs have stage games \( r_{\Theta_k} \), where \( \Theta_k \) are independently sampled from $\mathcal{P}$. We denote a generic (independent) copy of $\Theta_k$ by $\Theta$ and the vector of expected reward functions as $\mathbb{E}[r_\Theta]:=(\mathbb{E}[r^1_{\Theta^1}],\dots,\mathbb{E}[r^N_{\Theta^N}])$. To account for the uncertainty, we denote agent \( i \)'s action at stage \( k \) by \( A^i_k \) and its corresponding reward by $R^i_k = r^i_{\Theta^i_k}(A^i_k, A^{-i}_k).$

% As we highlighted in Section~\ref{sec:Intro}, a common approach to extending FP (or its smooth version) to model-free settings is to combine it with a Q-learning inspired reward estimation process \cite{Conv_Multi_Timescale_RL_Algos_in_NFGs_Leslie_Collins,sayin-et-al_fict-play-in-zero-sum-stoch-games,donmez-sayin_ind-q-learning-for-poly-games}. Although these methods can guarantee convergence, efficient exploration remains a major challenge, since the number of joint actions grows exponentially with $N$ (see Remark~\ref{rem:complexity}). 

A common approach to extending FP (or its smooth version) to model-free settings is to combine it with a Q-learning inspired reward estimation process \cite{Conv_Multi_Timescale_RL_Algos_in_NFGs_Leslie_Collins,sayin-et-al_fict-play-in-zero-sum-stoch-games,donmez-sayin_ind-q-learning-for-poly-games}. Although these methods can guarantee convergence, efficient exploration remains a major challenge, since the number of joint actions grows exponentially with $N$ (see Remark~\ref{rem:complexity}). 

To address this issue, we exploit anonymity and adopt an aggregative approach. Specifically, each agent maintains a Q-table that depends on the aggregate of others' actions, and uses agg-FP to update its beliefs and select actions. These Q-tables estimate the aggregate representation $\bar{r}$. For the Q-table updates, we follow an approach similar to that in \cite{Conv_Multi_Timescale_RL_Algos_in_NFGs_Leslie_Collins}, where the Q-values are updated at a faster rate than the beliefs. The resulting Q-table has the updates below for all $i\in[N]$, $k\geq 1$, $a^i\in\mathbb{A}^i$, and $x^{-i}\in\mathbb{X}$:
\begin{align*}
&\hat{Q}^i_{k+1}(a^i,x^{-i})\\
&=\hat{Q}^i_{k}(a^i,x^{-i}) + \mathds{1}_{\{A^i_k=a^i,~X^{-i}_k=x^{-i}\}}\beta_{\#_k(a^i,x^{-i})}(R^i_k-\hat{Q}^i_{k}(a^i,x^{-i}))).
\end{align*}
Here, $\#_k(a^i,x^{-i})$ is the number of times that $(a^i,x^{-i})$ is visited up to and including time $k$ and $(\beta_k)_{k\geq 0}$ is a step-size sequence that satisfies
\begin{align*}
\sum_{k\geq 0}\beta_k = \infty,\quad \sum_{k\geq 0}\beta_k^2 < \infty,\quad \lim_{k\to\infty}\frac{\alpha_k}{\beta_k}=0.
\end{align*}
Additionally, in this section, we require $(\alpha_k)_{k\geq 0}$ and  $(\beta_k)_{k\geq 0}$ to be non-increasing and satisfy $\sup_{k\geq 0}\alpha_{\lfloor \eta k\rfloor}/\alpha_k \leq A_\eta < \infty$ and $\sup_{k\geq 0} \beta_{\lfloor \eta k\rfloor}/\beta_k \leq B_\eta < \infty$ for all \(\eta \in (0,1)\). A valid choice is $\alpha_k=(k+1)^{-0.7}$ and $\beta_k=(k+1)^{-0.6}$. 

\begin{algorithm}[h!]
\caption{Two-Timescale Agg-FP}
\begin{algorithmic}
\State \textbf{Input:} Initial actions $a_0$, exploration probability $\delta$, step-size sequences $(\alpha_k)_{k\geq 0}$, $(\beta_k)_{k\geq 0}$
\State \textbf{Output:} Empirical action frequencies $(\hat{\gamma}_k)_{k\geq 0}$
% \vspace{1.5mm}
\State Initialize Q-table $Q_0 \gets 0$ and actions $A_0 \gets a_0$
% \vspace{1.5mm}
\For{each timestep $k=0,1,\dots$}
    % \vspace{1.5mm}
    \State Execute actions $A_k$ and receive rewards $R_k$
    % \vspace{1.5mm}
    \For{each agent $i=1,\dots,N$}
        \State Update Q-table $\forall a^i\in\mathbb{A},~\forall x^{-i}\in\mathbb{X}$:
        \begin{align*}
            &\hat{Q}^i_{k+1}(a^i,x^{-i})\\
            &+=\mathds{1}_{\{A^i_k=a^i,~X^{-i}_k=x^{-i}\}}\beta_{\#_k(a^i,x^{-i})}(R^i_k-\hat{Q}^i_{k}(a^i,x^{-i})))
        \end{align*}
        \State Update belief and action frequencies:
        \If{k = 0}
            \State $\hat{\mu}^i_k = \mathds{1}\{X^{-i}_0\},\quad \hat{\gamma}^i_k = \mathds{1}\{A^i_0\}$
        \Else
            \State $\hat{\mu}^i_{k} = \hat{\mu}^i_{k-1} + \alpha_{k}(\mathds{1}\{X^{-i}_{k}\}-\hat{\mu}^i_{k-1})$
            \State $\hat{\gamma}^i_{k} = \hat{\gamma}^i_{k-1} + \alpha_{k}(\mathds{1}\{A^i_k\}-\hat{\gamma}^i_{k-1})$
        \EndIf
    \EndFor
    \vspace{1.5mm}
    \State Generate random number $\omega \sim \text{Uniform}(0,1)$
    \If{$\omega < \delta$}
        \State Explore: $A^i_{k+1}\sim\text{Uniform}(\mathbb{A}),~\forall i\in[N]$
    \Else
        \State $A^{i}_{k+1} = \argmax_{a^i\in\mathbb{A}}\{\hat{Q}^i_{k+1}(a^i,\hat{\mu}^i_{k})\},~\forall i\in[N]$
    \EndIf

\EndFor

\State \Return $(\hat{\gamma}_k)_{k\geq 0}$
\end{algorithmic}
\label{algo}
\end{algorithm}

We summarize the overall algorithm, called two-timescale agg-FP, in Algorithm~\ref{algo}. Our main result is as follows.

\begin{theorem} \label{thm:main}
Assume that the agents repeatedly play an anonymous polymatrix game with random payoffs $\mathcal{G}_{\mathcal{P}}$ using two-timescale agg-FP. If $\mathbb{NE}_\epsilon(\mathcal{G}_{\mathcal{P}})$ is attractive under the $\delta$-greedy BR dynamics with reward functions $\mathbb{E}[r_\Theta]$, then $(\hat{\gamma}_k)_{k\geq 0}$ converges to $\mathbb{NE}_\epsilon(\mathcal{G}_{\mathcal{P}})$.
\end{theorem}

\begin{proof}
Similar to \cite{Conv_Multi_Timescale_RL_Algos_in_NFGs_Leslie_Collins,sayin-et-al_fict-play-in-zero-sum-stoch-games}, our approach is to employ two-timescale stochastic approximation techniques. Once we obtain the differential equation approximations for \(Q\) and \(\mu\), we will use Corollary~\ref{cor:same_conv_br_agg_br} and link the limit sets of these dynamics with those of the BR dynamics.

The Q-value of \((a^i, x^{-i})\) is updated only when that pair is observed. To accommodate this asynchronous scheme, we employ the results presented in \cite{perkins-leslie_asynch-stoch-approx-with-inclusion}. The two-timescale structure allows us to decouple the dynamics: the fast-varying \(Q\) treats \(\mu\) as effectively constant, and the slower \(\mu\) sees \(Q\) as effectively stabilized. Thus, following \cite{perkins-leslie_asynch-stoch-approx-with-inclusion}, the fast dynamics are given for all $i\in[N]$, $a^i\in\mathbb{A}$, and $x^{-i}\in\mathbb{X}$ by
\begin{align*}
\dot{q}^i_t(a^i,x^{-i}) = \mathbb{E}[\bar{r}^i_{\Theta^i}(a^i,x^{-i})] - q_t(a^i,x^{-i}),
\end{align*}
where $\Theta^i\sim\mathcal{P}^i$. This system has the globally asymptotically stable equilibrium $\mathbb{E}[\bar{r}^i_{\Theta^i}(a^i,x^{-i})]$. Hence, letting $a^j_{t,*}:= \argmax_{a^j\in\mathbb{A}}\{\mathbb{E}[\bar{R}^j_{\Theta^j}(a^j,\mu^{j}_t)]\}$ for all $j\in[N]$, the calibrated slow dynamics for each $i\in[N]$ is
\begin{align*}
% &\dot{\mu}^i_t = (1-\delta)\mathds{1}\left\{\sum_{j\neq i}\mathds{1}\left\{\argmax_{a^j\in\mathbb{A}}\{\mathbb{E}[\bar{R}^j_{\Theta^j}(a^j,\mu^{j}_t)]\}\right\}\right\} + \frac{\delta\mathbf{1}}{|\mathbb{X}|} - \mu^i_t,\\
&\dot{\mu}^i_t = (1-\delta)\mathds{1}\{\sigma(a^{-i}_{t,*})\} + \frac{\delta\mathbf{1}}{|\mathbb{X}|} - \mu^i_t,\\
&\dot{\gamma}^i_t = (1-\delta)\mathds{1}\{a^i_{t,*}\} + \frac{\delta\mathbf{1}}{n} - \gamma^i_t.
\end{align*}

The slow dynamics are exactly the agg-BR dynamics with the reward functions $\mathbb{E}[\bar{r}_{\Theta}]:=(\mathbb{E}[\bar{r}^1_{\Theta^1}],\dots,\mathbb{E}[\bar{r}^N_{\Theta^N}])$. Since $\mathcal{G}_{\mathcal{P}}$ is polymatrix, the matrix game with payoffs $\mathbb{E}[r_{\Theta}]$ is polymatrix. In addition, by assumption, $\mathbb{NE}_\epsilon(\mathcal{G}_\mathcal{P})$ is attractive under the $\delta$-greedy BR dynamics for this game. Therefore, using Corollary~\ref{cor:same_conv_br_agg_br}, we conclude that $\mathbb{NE}_\epsilon(\mathcal{G}_\mathcal{P})$ is attractive for the $\delta$-greedy agg-BR dynamics. Finally, \cite{perkins-leslie_asynch-stoch-approx-with-inclusion} ensures that $(\gamma_t)_{t\geq 0}$ and $(\hat{\gamma}_k)_{k\geq 0}$ share the same limit set.
\end{proof}

\begin{remark} \label{rem:eps_greedy_in_zero_sum_and_pot}
If \( \mathcal{G} \) is a zero-sum polymatrix game or a potential game, then for any \( \epsilon > 0 \), $\mathbb{NE}_\epsilon(\mathcal{G})$ is attractive for the \( \delta \)-greedy BR dynamics with sufficiently small \( \delta \). For these games, \( \mathbb{NE}(\mathcal{G}) \) admits a  Lyapunov function \( V \) \cite{monderer_shapley_pot-games,Fict_play_in_nets_Ewerhart_Valkanova} with a continuous (or upper hemicontinuous) gradient and the BR dynamics is upper hemicontinuous. Consequently, \( \nabla V^T(\pi) \dot{\pi} \) can always be made negative outside \( \mathbb{NE}_\epsilon(\mathcal{G}) \) for sufficiently small \( \delta \) (where $\dot{\pi}$ follows the $\delta$-greedy BR dynamics). This result paves the way for potential applications of two-timescale agg-FP, as we demonstrate for a specific game example in the next section. \qedwhite
\end{remark}

\section{Numerical Analysis} \label{sec:Numerical}

We illustrate our results using a 4-agent rock-paper-scissors (RPS) game \cite{fudenberg_levine_th-of-learning-in-games} with additive perturbations. At each stage $k$, each agent \( i \) plays a standard pairwise RPS game against every other agent and accumulates the reward
\begin{align}
r^i(a^i_k, a^{-i}_k) = \sum_{j=1,~j\neq i}^{N} \mathds{1}\{a^i_k\}^T  
\begin{bmatrix} 0 & -1 & 1 \\ 1 & 0 & -1 \\ -1 & 1 & 0 \end{bmatrix}  
\mathds{1}\{a^j_k\}. \label{eq:rps}
\end{align}
A random variable \( \Theta^i_k \) perturbs agent \( i \)'s reward, yielding the stage reward $r^i_{\theta^i_k}(a^i_k, a^{-i}_k) = r^i(a^i_k, a^{-i}_k) + \theta^i_k$, where $\theta^i_k$ is the realization of $\Theta^i_k$. For each $i\in[4]$ and $k\geq 0$, we assume that \( \Theta^i_k \) is sampled independently from the values \(\{-4, -2, 0, 2, 4\}\) with probabilities \(\mathcal{P}=(1/10, 1/5, 4/10, 1/5, 1/10)\). 

Let $\Theta^i$ be a copy of $\Theta^i_k$ for a generic $k$. For any realization \( \theta:=(\theta^1,\dots,\theta^4) \) of \( \Theta:=(\Theta^1,\dots,\Theta^4) \), the game with reward functions \( r_{\theta}:=(r^1_{\theta^1},\dots,r^4_{\theta^4})\) is polymatrix and anonymous. Moreover, \( \mathbb{E}[r^i_{\Theta^i}] = r^i \) for each $i\in[4]$, where $(r^1,\dots,r^4)$ as specified by \eqref{eq:rps}, defines a zero-sum\footnote{A multi-player matrix game is zero-sum if it satisfies $\sum_{i=1}^N r^i(a^i, a^{-i}) = 0$ for every action profile $a$ \cite{donmez-sayin_ind-q-learning-for-poly-games, Fict_play_in_nets_Ewerhart_Valkanova}. 
% This is not to be confused with the common interpretation of two-player zero-sum games, where one agent aims to maximize and the other to minimize a reward function \cite{basar-olsder_dyn-noncoop-game-theory}.
} polymatrix game. Therefore, Theorem~\ref{thm:main} and Remark~\ref{rem:eps_greedy_in_zero_sum_and_pot} imply that, given any \( \epsilon > 0 \), two-timescale agg-FP converges to \( \mathbb{NE}_\epsilon(\mathcal{G}_{\mathcal{P}}) \) for sufficiently small exploration probability \( \delta \). Note that \( \mathbb{NE}(\mathcal{G}_{\mathcal{P}}) = \{(1/3,1/3,1/3,1/3),\dots,(1/3,1/3,1/3,1/3)\)\}, i.e. each agent selects each action with probability \( 1/3 \).

We simulate two-timescale agg-FP using this game and parameters \( \delta=0.1 \), \( \beta_k=(k+1)^{-0.6} \), and \( \alpha_k=(k+1)^{-0.7} \). In Fig.~\ref{fig:emp_freqs}, we present only the empirical action frequencies of agent 1, though other agents exhibit similar behavior. These trajectories confirm that \( (\hat{\gamma}_k)_{k\geq 0} \) approaches to a neighborhood of \( \mathbb{NE}(\mathcal{G}_{\mathcal{P}}) \).

Figs.~\ref{fig:q_errors} and~\ref{fig:ne_errors} compare two-timescale agg-FP with two-timescale FP and individual Q-learning\footnote{Two-timescale FP uses the same parameters as two-timescale agg-FP. Individual Q-learning shares the same Q-update rate, but uses $\delta$ as temperature. We also use the same perturbation realization $(\theta^i_k)_{i\in[4],k\geq 0}$ in each algorithm.}. Figure~\ref{fig:q_errors} shows the \( l_1 \)-error between \( \mathbb{E}[r_{\Theta}] \) and the Q-table from two-timescale FP, as well as the \( l_1 \)-error between \( \mathbb{E}[\bar{r}_{\Theta}] \) and the Q-table from two-timescale agg-FP. The Q-error from agg-FP converges faster, supporting the intuition that reducing the joint action space improves exploration. The Q-error from individual Q-learning is not shown, as the \( l_1 \)-distance between the Q-table and the expected reward functions is not well-defined in this case. Fig.~\ref{fig:ne_errors} displays the $l_1$-distance of empirical action frequencies to $\mathbb{NE}(\mathcal{G}_{\mathcal{P}})$, confirming that agg-FP converges faster. Reducing the Q-table size also appears to smoothen the learning process.

\begin{remark}
Complex games---with many agents, large action spaces, or multiple equilibria---may require more aggressive exploration (i.e., larger $\delta$) to ensure convergence in a timely manner. Conversely, in simpler environments, a smaller $\delta$ may suffice, yielding convergence to a better approximation of a Nash equilibrium of the game.
\qedwhite
\end{remark}

\begin{remark}
While using agg-FP instead of FP reduces the size of the joint action space and lowers memory requirements, it also leads to a loss of agent-specific information. In particular, by aggregating the agents’ actions, agg-FP discards which agent played which action. Nonetheless, as long as the game is polymatrix and anonymous, and the goal is to learn a Nash equilibrium (or an $\epsilon$-Nash equilibrium), this compression does not negatively impact performance.
\qedwhite
\end{remark}

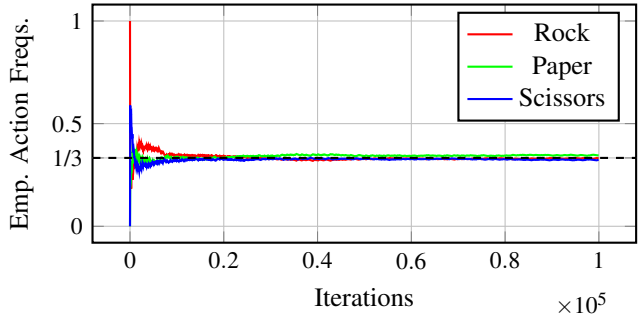
\begin{figure}[h!]
    \centering
    % First Plot
    \begin{tikzpicture}
        \begin{axis}[
            xlabel={Iterations},
            ylabel={Emp. Action Freqs.},
            grid=major,
            extra y ticks={0.333},  % Use decimal for positioning
            extra y tick labels={1/3},  % Display as fraction
            width=1.02\columnwidth,
            height=0.55*\columnwidth,
            legend pos=north east,
            xlabel near ticks, ylabel near ticks,
            tick label style={font=\small},
            enlargelimits=0.08,
            tick scale binop=\times,
            every axis/.append style={line width=0.8pt}
        ]
            \addplot[red, thick] table [col sep=comma, x=x, y=y] {Plot_Data/agg_act_freqs_hist_action_0.csv};
            \addlegendentry{Rock};
            \addplot[green, thick] table [col sep=comma, x=x, y=y] {Plot_Data/agg_act_freqs_hist_action_1.csv};
            \addlegendentry{Paper};
            \addplot[blue, thick] table [col sep=comma, x=x, y=y] {Plot_Data/agg_act_freqs_hist_action_2.csv};
            \addlegendentry{Scissors};
            \draw[dashed, black, thick] (axis cs:\pgfkeysvalueof{/pgfplots/xmin},0.333) -- (axis cs:\pgfkeysvalueof{/pgfplots/xmax},0.333);
        \end{axis}
    \end{tikzpicture}
    \caption{Empirical action frequencies of agent 1 from agg-FP.}
    \label{fig:emp_freqs}
\end{figure}

\begin{figure}[h!]
    \centering
    \begin{tikzpicture}
        \begin{axis}[
            xlabel={Iterations},
            ylabel={Q-Table Error},
            grid=major,
            width=1.02\columnwidth,
            height=0.55*\columnwidth,
            legend pos=north east,
            xlabel near ticks, ylabel near ticks,
            tick label style={font=\small},
            enlargelimits=0.05,
            tick scale binop=\times,
            every axis/.append style={line width=0.8pt}
        ]
            \addplot[blue, thick] table [col sep=comma, x=x, y=y] {Plot_Data/std_acc_Q_diffs.csv};
            \addlegendentry{FP};
            \addplot[red, thick] table [col sep=comma, x=x, y=y] {Plot_Data/agg_acc_Q_diffs.csv};
            \addlegendentry{agg-FP};
        \end{axis}
    \end{tikzpicture}
    \caption{Total Q-error across all agents and action profiles.}
    \label{fig:q_errors}
\end{figure}
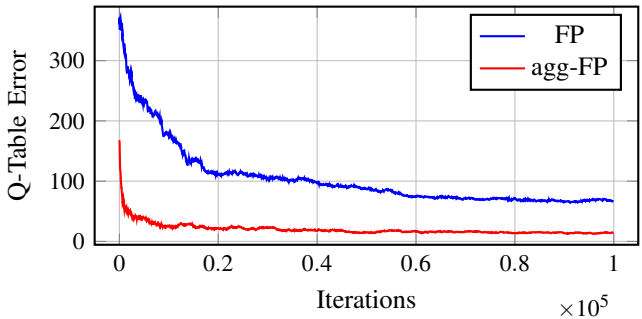

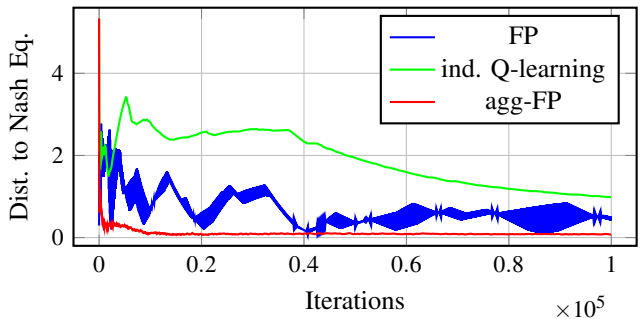
\begin{figure}[h!]
    \centering
    \begin{tikzpicture}
        \begin{axis}[
            xlabel={Iterations},
            ylabel={Dist. to Nash Eq.},
            grid=major,
            width=1.05\columnwidth,
            height=0.55*\columnwidth,
            legend pos=north east,
            xlabel near ticks, ylabel near ticks,
            tick label style={font=\small},
            enlargelimits=0.05,
            tick scale binop=\times,
            every axis/.append style={line width=0.8pt}
        ]
            \addplot[blue, thick] table [col sep=comma, x=x, y=y] {Plot_Data/std_acc_strat_diffs.csv};
            \addlegendentry{FP};
            \addplot[green, thick] table [col sep=comma, x=x, y=y] {Plot_Data/indQ_acc_strat_diffs.csv};
            \addlegendentry{ind. Q-learning};
            \addplot[red, thick] table [col sep=comma, x=x, y=y] {Plot_Data/agg_acc_strat_diffs.csv};
            \addlegendentry{agg-FP};
        \end{axis}
    \end{tikzpicture}
    \caption{Total $l_1$ distance to Nash equilibrium across all agents.}
    \label{fig:ne_errors}
\end{figure}

\section{Conclusion and Acknowledgment}

We have studied learning Nash equilibria in anonymous games through repeated play. We have introduced agg-FP and proved that it converges under the same conditions as FP for anonymous polymatrix games. We have then extended agg-FP to model-free settings and games with random payoffs by combining it with a reward estimation process, resulting in two-timescale agg-FP. For anonymous polymatrix games with random payoffs, we have proved that two-timescale agg-FP converges to a neighborhood of a Nash equilibrium whenever the \( \delta \)-greedy BR dynamics converge for the expected reward functions. We have presented simulation results demonstrating that two-timescale agg-FP accelerates convergence.

The authors acknowledge, with thanks, fruitful discussions with Kaiqing Zhang during different phases of the research reported here.

\bibliographystyle{IEEEtran}
\bibliography{Kara_Refs}

\end{document}